\documentclass[11pt]{article}

\usepackage{amsmath,amssymb,amsfonts,amsthm,xspace,graphicx,relsize,bm,mathtools,xcolor,mathrsfs}
\usepackage[pagebackref]{hyperref}
\usepackage{booktabs}
\def\01{\{0,1\}}

\newcommand{\eps}{\varepsilon}
\usepackage{braket} 
\newcommand{\ketbra}[2]{|#1\rangle\langle#2|}
\newcommand{\inpc}[2]{\langle{#1}|{#2}\rangle} 
\newcommand{\tr}{\mbox{\rm Tr}}
\newcommand{\err}{\mbox{\rm err}}

\newcommand{\norm}[1]{\mbox{$\parallel{#1}\parallel$}}
\newcommand{\Cc}{{\mathcal C}} 
\newcommand{\Exp}{\mathbb{E}}

\newcommand{\Ind}{\ensuremath{\mathbf{1}}}
\newcommand{\Id}{\ensuremath{\mathsf{Id}}}

\newtheorem{definition}{Definition}
\newtheorem{theorem}{Theorem}
\newtheorem{lemma}[theorem]{Lemma}
\newtheorem{proposition}[theorem]{Proposition}
\newtheorem{corollary}[theorem]{Corollary}

\def\01{\{0,1\}}

\DeclareMathOperator{\spann}{span}

\usepackage{kpfonts}

\usepackage[margin=1in]{geometry}
\hypersetup{
    colorlinks,
    linkcolor={blue!100!black},
    citecolor={blue!100!black},
  }

\let\oldproofname=\proofname
\renewcommand{\proofname}{\rm\bf{\oldproofname}}

\def\makeSkob#1#2#3{%
\def\LLL{\mathopen{}\mathclose\bgroup\left} \def\RRR{\aftergroup\egroup\right}
\expandafter \edef \csname #1\endcsname #2##1#3{\SkobInner}
\def\LLL{} \def\RRR{}
\expandafter \edef \csname #1O\endcsname #2##1#3{\SkobInner}
\def\LLL{\bigl} \def\RRR{\bigr}
\expandafter \edef \csname #1A\endcsname #2##1#3{\SkobInner}
\def\LLL{\Bigl} \def\RRR{\Bigr}
\expandafter \edef \csname #1B\endcsname #2##1#3{\SkobInner}
\def\LLL{\biggl} \def\RRR{\biggr}
\expandafter \edef \csname #1C\endcsname #2##1#3{\SkobInner}
\def\LLL{\Biggl} \def\RRR{\Biggr}
\expandafter \edef \csname #1D\endcsname #2##1#3{\SkobInner}
}

\def \elem[#1]{(#1)}
\def\SkobInner{\LLL(##1\RRR)} \makeSkob{s}[]
\def\SkobInner{\LLL[##1\RRR]} \makeSkob{sk}[]
\def\SkobInner{\LLL\lbrace##1\RRR\rbrace} \makeSkob{sfig}{}{}
\def\SkobInner{\LLL\lfloor##1\RRR\rfloor} \makeSkob{floor}[]
\def\SkobInner{\LLL\lceil##1\RRR\rceil} \makeSkob{ceil}[]
\def\SkobInner{\LLL\langle##1\RRR\rangle} \makeSkob{ip}<>
\def\SkobInner{\LLL\lvert##1\RRR\rvert} \makeSkob{abs}||
\def\SkobInner{\LLL\lVert##1\RRR\rVert} \makeSkob{norm}||
\def\SkobInner{\LLL\lVert##1\RRR\rVert_{\noexpand\mathrm F}} \makeSkob{normFrob}||
\def\SkobInner{\LLL\lVert##1\RRR\rVert_{\noexpand\mathrm{tr}}} \makeSkob{normtr}||
\newcommand{\qqand}{\qquad\text{and}\qquad}

\newcommand{\Johnson}[2]{\mathcal{J}(#1,#2)}
 \DeclareMathOperator{\Var}{Var}

\begin{document}

\title{Quantum Coupon Collector}
\author{Srinivasan Arunachalam\thanks{IBM Research. Part of this work was done while a PhD student at CWI supported by ERC Consolidator Grant 615307-QPROGRESS, and a postdoc at MIT funded by the MIT-IBM Watson AI Lab under the project \emph{Machine learning in Hilbert space}. 
 {\tt Srinivasan.Arunachalam@ibm.com}}
 \and
Aleksandrs Belovs\thanks{Faculty of Computing, University of Latvia.  Supported by the ERDF project number 1.1.1.2/I/16/113. {\tt aleksandrs.belovs@lu.lv}}
 \and
Andrew M.\ Childs\thanks{Department of Computer Science, Institute for Advanced Computer Studies, and Joint Center for Quantum Information and Computer Science, University of Maryland. Supported by the Army Research Office (grant W911NF-20-1-0015); the Department of Energy, Office of Science, Office of Advanced Scientific Computing Research, Quantum Algorithms Teams and Accelerated Research in Quantum Computing programs; and the National Science Foundation (grant CCF-1813814). {\tt amchilds@umd.edu}}
 \quad \and 
Robin Kothari\thanks{Microsoft Quantum and Microsoft Research. {\tt robin.kothari@microsoft.com}}
 \and
Ansis Rosmanis\thanks{Graduate School of Mathematics, Nagoya University, Japan. Supported by the JSPS International Research Fellowship program and by the JSPS KAKENHI Grant Number JP19F19079. {\tt ansis.rosmanis@math.nagoya-u.ac.jp}}
 \and
 Ronald de Wolf\thanks{QuSoft, CWI and University of Amsterdam, the Netherlands. Partially supported by ERC Consolidator Grant 615307-QPROGRESS (which ended Feb 2019), and by the Dutch Research Council (NWO) through Gravitation-grant Quantum Software Consortium 024.003.037, and QuantERA project QuantAlgo 680-91-034. {\tt rdewolf@cwi.nl}}
}
\maketitle

\begin{abstract}
We study how efficiently a $k$-element set $S\subseteq[n]$ can be learned from a uniform superposition $\ket{S}$ of its elements.
One can think of $\ket{S}=\sum_{i\in S}\ket{i}/\sqrt{|S|}$ as the quantum version of a uniformly random sample over $S$, as in the classical analysis of the ``coupon collector problem.'' We show that if $k$ is close to $n$, then we can learn $S$ using asymptotically fewer quantum samples than random samples. In particular, if there are $n-k=O(1)$ missing elements then $O(k)$ copies of $\ket{S}$ suffice, in contrast to the $\Theta(k\log k)$ random samples needed by a classical coupon collector. On the other hand, if $n-k=\Omega(k)$, then $\Omega(k\log k)$ quantum samples are~necessary.

More generally, we give tight bounds on the number of quantum samples needed for every~$k$ and $n$, and we give efficient quantum learning algorithms. We also give tight bounds in the model where we can additionally reflect through $\ket{S}$. Finally, we relate coupon collection to a known example separating proper and improper PAC learning that turns out to show no separation in the quantum case.
\end{abstract}
\section{Introduction}

Learning from quantum states is a major topic in quantum machine learning. While this task has been studied extensively~\cite{bshouty:quantumpac,servedio&gortler:equivalencequantumclassical,atici&servedio:qlearning,atici&servedio:testing,arunachalam:optimalpaclearning,grilo:LWEeasy,arunachalam:qexactlearning}, many fundamental questions about the power of quantum learning remain. Determining properties of quantum states has potential applications not only in the context of machine learning, but also as a basic primitive for other types of quantum algorithms and for quantum information processing more~generally.

In this paper we study a very simple and natural quantum learning problem. We are given copies of the uniform superposition 
$$
\ket{S}:=\frac{1}{\sqrt{|S|}}\sum_{i\in S}\ket{i} 
$$
over the elements of an unknown set $S\subseteq [n] := \{1,\ldots,n\}$ (sometimes referred to as a uniform \emph{quantum sample}
from $S$~\cite{ATS03}). Assume we know the size $k := |S| < n$. Our goal is to learn~$S$ exactly. How many copies of $\ket{S}$ do we need for this? And given the information-theoretically minimal number of copies needed, can we learn $S$ gate-efficiently (i.e., using a quantum circuit with gate count polynomial in $k$ and $\log n$)?

As a warm-up, first consider what happens if we just measure our copies of $\ket{S}$ in the computational basis, giving uniform samples from~$S$.
How many such samples do we need before we learn $S$? As long as there is some element of~$S$ that we have not seen, we cannot even guess $S$ with constant success probability, so we need to sample until we see all $k$ distinct elements. This is known as the ``coupon collector problem.'' Analyzing the required number of samples is easy to do in expectation, as follows. Suppose we have already seen $i<k$ distinct elements from~$S$. Then the probability that we see a new element in the next sample is $(k-i)/k$, and the expected number of samples to see an $(i+1)$st element is the reciprocal of that probability, $k/(k-i)$. By linearity of expectation we can add this up over all $i$ from $0$ to $k-1$, obtaining the expected number of samples to see all $k$ elements:
$$
\sum_{i=0}^{k-1}\frac{k}{k-i}=k\sum_{j=1}^k\frac{1}{j}\sim k\ln k.
$$
With a bit more work one can show that $\Theta(k\log k)$ samples are necessary and sufficient to identify $S$ with high probability \cite[Chapter~3.6]{motwani1995randomized}:

\begin{proposition}[Classical coupon collector]\label{prop:classical}
Given uniformly random samples from a set $S\subseteq[n]$ of size $k<n$, the number of samples needed to identify $S$ with high probability is $\Theta(k \log k)$.
\end{proposition}

The relationship between the probability of seeing all elements of $S$ and the number of samples is extremely well understood. In particular, we can achieve probability arbitrarily close to $1$ using only $k \ln k + O(k)$ samples~\cite[Theorem 3.8]{motwani1995randomized}.

Of course, measuring $\ket{S}$ in the computational basis is not the only approach a quantum computer could take. The goal of this paper is to identify when and how we can do better, reducing the number of copies of~$\ket{S}$ that are used to solve this ``quantum coupon collector problem.'' It turns out that we can asymptotically beat the classical threshold of $\Theta(k\log k)$ if and only if the number $m=n-k$ of ``missing elements'' is small (whereas classically the parameter~$m$ is irrelevant). Specifically, we give a simple, gate-efficient quantum algorithm that learns $S$ from $O(n\log(m+1))$ copies of~$\ket{S}$. For small $m$ this is significantly more efficient than classical coupon collection. In particular, for $m=O(1)$ we only need $O(k)$ quantum samples, saving a factor of $O(\log k)$.

As we explain in Section~\ref{sec:properPAC}, this result is relevant for the comparison of \emph{proper} and \emph{improper} learning in the PAC model. A ``proper'' learner is one that only outputs hypotheses from the same concept class that its target function comes from. The coupon collector problem can be viewed as a learning task where the sample complexity of proper learners from classical random examples is asymptotically higher than that of proper learners from quantum examples.

We also prove \emph{lower} bounds on the number~$T$ of copies needed, using the general (i.e., negative-weights) adversary bound of quantum query complexity~\cite{hls:madv}.
This approach may be surprising, since no queries are involved when trying to learn $S$ from copies of $\ket{S}$.
However, the adversary bound also characterizes the quantum query complexity of ``state conversion''~\cite{lmrss:stateconv} and ``state discrimination.'' Our learning problem may be viewed as the problem of converting the state $\ket{S}^{\otimes T}$ to a basis state that gives a classical description of the $k$-set $S$.
To employ the general adversary bound, we exploit the underlying symmetries of the problem using the mathematical machinery of association schemes (see also \cite{AMRR11,LR20} for prior uses of association schemes in proving adversary lower bounds).
Using this, we show that, unless the number of missing elements $m=n-k$ is very small, the $O(k\log k)$ classical coupon collector algorithm is essentially optimal even in the quantum case. This means that the quantum coupon collector might as well just measure the copies of the state in the computational basis, unless $m$ is very small.

We also study the situation where, in addition to copies of $\ket{S}$, we can also apply a unitary operation $R_S=2\ketbra{S}{S}-\Id$ that reflects through the state~$\ket{S}$ (i.e., $R_S\ket{S}=\ket{S}$ and $R_S\ket{\phi}=-\ket{\phi}$ for all states $\ket{\phi}$ orthogonal to $\ket{S}$). This model is reasonable to consider because if we had a unitary that prepared $\ket{S}$, or even $\ket{S}\ket{\psi}$ for some garbage state $\ket{\psi}$, starting from some canonical state~$\ket{0}$, then we could use this unitary to create the unitary $R_S$ in a black-box manner. For example, if $U\ket{0}=\ket{S}$, then $R_S = U(2\ketbra{0}{0}-\Id)U^\dagger$.

This model gives us extra power and enables more efficient learning of the set~$S$: $\Theta(\sqrt{km})$ states and reflections are necessary and sufficient to learn $S$ for large $k$ (i.e., small $m$), and $\Theta(k)$ states and reflections are necessary and sufficient for small $k$.

The following table summarizes our main results. Sections~\ref{sec:upperbound} and~\ref{sec:lowerbound} prove the upper and lower bounds in the first row, respectively, while Section~\ref{sec:numberofreflections} proves the results in the second row.

\renewcommand{\arraystretch}{1.2}
\renewcommand{\tabcolsep}{10pt}

\begin{table}[ht]
\centering
\begin{tabular}[t]{p{8em}p{13em}p{13em}}
\toprule
& $k\geq n/2$ & $k\leq n/2$ \\
\midrule
Number of copies of $\ket{S}$: 
& $\Theta(k\log(m+1))$ 
\newline Theorem \ref{thm:uppersample} and Theorem \ref{thm:lowersample} 
& $\Theta(k\log k)$ 
\newline Proposition \ref{prop:classical} and Theorem \ref{thm:lowersample}\\[1ex]
Number of copies and reflections:
& $\Theta(\sqrt{km})$ \newline
Theorem \ref{thm:uppersamplerefsmallm} and Theorem \ref{thm:lowersamplerefsmallm}
& $\Theta(k)$ \newline
Theorem \ref{thm:uppersamplerefsmallk} and Theorem \ref{thm:lowersamplerefsmallk}\\
\bottomrule
\end{tabular}
\caption{Main results about the complexity of learning the set $S$ with $m=n-k$ missing elements}
\label{tab:mainresults}
\end{table}%
\renewcommand{\arraystretch}{1}

We contrast our work with recent results on the quantum query complexity of approximate counting by Aaronson, Kothari, Kretschmer, and Thaler~\cite{akkt:laurent}. They consider a similar model, given copies of the state $\ket{S}$, the ability to reflect through $\ket{S}$, and also the ability to query membership in~$S$. However, in their work the size of $S$ is unknown and the goal is to approximately \emph{count} this set up to small multiplicative error. They obtain tight upper and lower bounds on the complexity of this approximate-counting task using techniques quite different from ours (specifically, Laurent polynomials for the lower bounds). This allows them to give an oracle separation between the complexity classes SBP and QMA.
In contrast, in our case the size~$k$ of the set~$S$ is already known to the learner from the start, and the goal is to \emph{identify}~$S$ exactly.

\section{Upper bound on quantum samples}\label{sec:upperbound}

In this section we prove upper bounds on the number of copies of $\ket{S}$ that suffice to identify the $k$-element set $S\subseteq[n]$ with high probability.

The easiest way to recover $S$ is by measuring $O(k\log k)$ copies of $\ket{S}$ in the computational basis. By the classical coupon collector problem (Proposition \ref{prop:classical}), we will (with high probability) see all elements of $S$ at least once. As we will show later, this number of copies of $\ket{S}$ turns out to be asymptotically optimal if the number of missing elements $m=n-k$ is large (at least polynomial in~$n$). However, here we show that something better is possible for very small~$m$.

\begin{theorem}[Upper bound for small $m$]\label{thm:uppersample}
Let $S\subseteq [n]$ be a set of size $k<n$ and let $m=n-k$. We can identify $S$ with high probability using $O(k \log (m+1))$ copies of $\ket{S}$ by a gate-efficient quantum~algorithm.
\end{theorem}

\begin{proof}
This bound is trivial when $m$ is polynomial in $n$, since an upper bound of $O(k \log k)$ follows from Proposition \ref{prop:classical}. So let us now assume that $m\leq n^{1/4}$ and hence $k\geq n-n^{1/4}$.

Consider the uniform superposition over all elements of the universe $[n]$:
$$
\Ket{[n]}=\frac{1}{\sqrt{n}}\sum_{i\in[n]}\ket{i}.
$$
Performing the 2-outcome projective measurement with operators $\ket{[n]}\bra{[n]}$ and $\Id-\ketbra{[n]}{[n]}$ is no harder than preparing $\ket{[n]}$,
so it can be implemented gate-efficiently.
If we apply this measurement to a copy of $\ket{S}$, then we get the first outcome with probability $|\inpc{S}{[n]}|^2=k/n$ and the second outcome with probability~$m/n$. In the latter case, the post-measurement state is
$$
\ket{\psi}=\sqrt{\frac{m}{n}}\ket{S}-\sqrt{\frac{k}{n}}\ket{\overline{S}}
$$
which is close to $-\ket{\overline{S}}$ if $m\ll n$.

We use an expected number of $\smash{O\bigl(\frac{n}{m}\cdot m\log(m+1)\bigr)}=O(n\log(m+1))$ copies of $\ket{S}$ to prepare $O(m\log(m+1))$ copies of $\ket{\psi}$.
If $\ket{\psi}$ were exactly equal to $-\ket{\overline{S}}$, then measuring in the computational basis would sample uniformly over the set $\overline{S}$ of $m$ missing elements, and $O(m\log(m+1))$ such samples suffice to recover $\overline{S}$ by the classical coupon collector problem (Proposition \ref{prop:classical}).
Instead,~$\ket{\psi}$ only approximately equals $-\ket{\overline{S}}$: if we measure it then each $i\in\overline{S}$ has probability $\frac{k}{nm}$, while each $i\in S$ has (much smaller but nonzero) probability $\frac{m}{nk}$. Suppose we prepare and measure $T=10m\log(m+1)$ copies of $\ket{\psi}$.
Then the expected number of occurrences of each $i\in\overline{S}$ is $T\cdot \frac{k}{nm}\geq 5\log(m+1)$ since $k\geq n/2$, while the expected number of occurrences of each $i\in S$ is $T\cdot \frac{m}{nk}=O(\log(n)/n^{3/2})$.
In both cases the number of occurrences is tightly concentrated.\footnote{Suppose we flip $T$ 0/1-valued coins, each taking value 1 with probability~$p$. Let $X$ be their sum (i.e., the number of 1s), which has expectation $\mu=pT$. The Chernoff bound implies $\Pr[X\leq(1-\delta)\mu]\leq \exp(-\delta^2\mu/2)$. To get concentration for the number of occurrences of a specific $i\in\overline{S}$, apply this tail bound with $p=k/(nm)$, $\mu=Tp\geq 5\log(m+1)$, $\delta=4/5$ to obtain $\Pr[X\leq \log(m+1)]\ll 1/m$. Hence, by a union bound, the probability that among the $m$ elements $i\in\overline{S}$ there is one of which we see fewer than $\log(m+1)$ occurrences, is $\ll 1$.
For an $i\in S$, by Markov's inequality the probability to see at least $\log(m+1)$ occurrences of this~$i$ among the $T$ samples is $\ll 1/n$, and we can use a union bound over all $i\in S$.} 
Hence if we keep only the elements that appear, say, at least $\log(m+1)$ times among the $T$ outcomes, then with high probability we will have found $\overline{S}$, and hence learned $S=[n]\setminus \overline{S}$. 
\end{proof}

\section{Lower bound on quantum samples}\label{sec:lowerbound}

In this section we prove lower bounds on the number of copies of $\ket{S}$ needed to identify $S$ with high probability.
Before establishing the lower bounds claimed in Table \ref{tab:mainresults}, we introduce some preliminary concepts, namely the $\gamma_2$-norm (Section~\ref{ssec:gamma2norm}), association schemes (Section~\ref{ssec:assoc}), the Johnson scheme (Section~\ref{ssec:johnson}), and the adversary bound for state discrimination (Section~\ref{ssec:advLB}). The lower bound itself is established in Section~\ref{sec:lowersample}.


\subsection{\texorpdfstring{$\gamma_2$}{Gamma2}-norm}\label{ssec:gamma2norm}
 The $\gamma_2$-norm of a $D_1\times D_2$ matrix $A$ with entries $A\elem[x,y]$ for $x \in [D_1]$ and $y \in [D_2]$ can be defined in two equivalent ways~\cite[Section~3]{belovs:phd}. 
The primal definition is 
\begin{alignat}{2}
\label{eq:primalgamma2}
\begin{aligned}
&\mbox{\rm minimise} &\quad& \max \sfigB{ \max\nolimits_{x\in [D_1]} \norm|u_x|^2, \max\nolimits_{y\in [D_2]} \norm|v_y|^2 } \\
& \mbox{\rm subject to}&&  
A\elem[x,y] = \braket{u_x, v_y} \qquad \text{\rm for all $x\in [D_1]$ and $y\in [D_2]$},
\end{aligned}
\end{alignat} 
where $\{u_x : x \in [D_1]\}$ and $\{v_y: y \in [D_2]\}$ are vectors of the same dimension.
The dual definition is
\begin{alignat}{2}
\begin{aligned}
\label{eq:dualgamma2}
&\mbox{\rm maximise} &\quad& \| \Gamma\circ A\| \\
& \mbox{\rm subject to}&&  \|\Gamma\|\le 1
\end{aligned}
\end{alignat}
where $\Gamma$ ranges over $D_1\times D_2$-matrices, $\circ$ denotes the Hadamard (entrywise) product of matrices, and $\|\cdot \|$ is the spectral norm of a matrix.
Note that $\gamma_2(A\circ B)\leq\gamma_2(A)\gamma_2(B)$: consider the vectors obtained from the  optimal feasible solutions of $\gamma_2(A)$ and $\gamma_2(B)$ in Eq.~\eqref{eq:primalgamma2} and observe that the tensor product of these vectors forms a feasible solution for the primal problem for $\gamma_2(A\circ B)$ with (not necessarily minimal) value $\gamma_2(A)\gamma_2(B)$.

\subsection{Association schemes} \label{ssec:assoc}

Here we present a quick introduction to association schemes (see, for example, \cite[Chapter~1]{godsil:assoc2018} for a more thorough treatment).
\begin{definition}
An \emph{association scheme} on the set $U$ is a finite set of real symmetric $U\times U$ matrices $\{A_0,A_1,\dots, A_s\}$ satisfying all the following properties:
\begin{itemize}
\item each $A_j$ only has entries 0 and 1;
\item $A_0$ is the identity matrix $\Id$;
\item $\sum_{j=0}^s A_j$ is the all-1 matrix~$J$; and
\item for every $i$ and $j$, the product $A_i A_j$ is a linear combination of the matrices $\{A_0,\dots,A_s\}$.
\end{itemize}
The space spanned by the set $\{A_0,A_1,\dots, A_s\}$ forms an algebra, which is called the \emph{Bose--Mesner algebra} corresponding to the scheme. By abuse of terminology, we may also refer to this algebra as the association~scheme.
\end{definition}

We now state a few properties of $\{A_0,\ldots,A_s\}$.  
First, observe that $A_j$ has zero diagonal for $j>0$.  Additionally, $\{A_0,\ldots,A_s\}$ form a basis of the corresponding Bose--Mesner algebra, since for every $(x,y)$, there is exactly one $j$ for which $A_j(x,y)\ne 0$. Also, the basis $\{A_0,\dots,A_s\}$ satisfies
$
A_i\circ A_j = \Ind_{[i = j]} A_i
$,
where $\Ind_{[P]}$ is the indicator function of predicate~$P$ (i.e., $1$ if $P$ is true and $0$ if $P$ is false).
It is possible to find another basis $\{E_0,\dots,E_s\}$  consisting of  \emph{idempotent} matrices for $\spann\{A_0,\dots,A_s\}$ that satisfy 
$
E_iE_j = \Ind_{[i = j]} E_i,
$
 with respect to the usual product of matrices. The operators $E_i$ are orthogonal projectors onto the \emph{eigenspaces} of the association scheme.
We have 
\[
E_0 = J/N
\qqand
\sum_{j=0}^s E_j = \Id,
\]
where $N = |U|$.
Since both $\{A_i\}$ and $\{E_j\}$ are bases for the space of $N\times N$ matrices, it is possible to~write
\begin{align}
    \label{eq:eigenvaluesanddual}
A_i = \sum_{j=0}^s p_i(j)E_j
\quad\text{and}\quad
E_j = \sum_{i=0}^s \frac{q_j(i)}{N}A_i,
\end{align}
where $p_i(j)$ and $q_j(i)$ are called the \emph{eigenvalues} and \emph{dual eigenvalues} of the association scheme, respectively.

It is easy to show that the Hadamard product and the usual product of any two elements of the association scheme also belong to the association scheme. Clearly for every $i,j$ we know that $A_i\circ A_j$ and $E_i\cdot E_j$ are elements of the basis of the scheme. Also observe that $A_i\cdot A_j$ and $E_i\circ E_j$ are elements of the scheme by writing out these products using Eq.~\eqref{eq:eigenvaluesanddual} and observing that $A_i\cdot A_j$ (resp.~$E_i\circ E_j$) is a linear combination of elements of $\{A_0,\ldots,A_s\}$ (resp.~$\{E_0,\ldots,E_s\}$). In particular, we can write 
\begin{align}
\label{eq:hadamardproductE}
E_i\circ E_{j} = \frac{1}{N}\sum_{\ell=0}^s q_{i,j}(\ell) E_{\ell}.
\end{align}
The real numbers $q_{i,j}(\ell) $ are called the \emph{Krein parameters} of the association scheme.

\subsection{Johnson scheme}
\label{ssec:johnson}

In the Johnson association scheme $\Johnson{n}{k}$, the set $U$ is the set of all $k$-subsets of $[n]$.
Therefore, $N=|U|=\binom{n}{k}$.
Let $m = \min\{k,n-k\}$.
For $j=0,1,\dots,m$, define $A_j\elem[x,y]:=\Ind_{[|x\cap y|=k-j]}$. The idempotent $E_j$ is defined as follows:
for $x\in U$, let $e_x\in\mathbb{R}^U$ be the indicator vector defined as $e_x(y)=\Ind_{[x=y]}$ for $y\in U$, and let
\[
\mathcal{V}_j := \begin{cases}
\mathrm{span}\big\{\sum_{x\supseteq z}e_x\colon z\subseteq[n] \text{ with } |z|=j\big\}
& \text{if } k\le n/2,\\
\mathrm{span}\big\{\sum_{x\subseteq z}e_x\colon z\subseteq[n] \text{ with } |z|=n-j\big\}
& \text{if } k> n/2,\\
\end{cases}
\]
where the sums are over $x\in U$. These spaces satisfy  $\mathcal{V}_0\subset\mathcal{V}_1\subset\cdots\subset\mathcal{V}_m=\mathbb{R}^U$ and the dimension of $\mathcal{V}_j$ is $\binom{n}{j}$.
For $j\in\{1,\ldots,m\}$, the idempotent $E_j$ is defined as the orthogonal projector on $\mathcal{V}_j\cap \mathcal{V}_{j-1}^{\perp}$, and $E_0$ is the orthogonal projector on $\mathcal{V}_0$.
Hence, for $j\in\{0,1,\ldots,m\}$, the dimension of the $j$th eigenspace is
\begin{align}
    \label{eq:defnoftraceEj}
d_j := \tr[E_j] = \binom{n}{j} - \binom{n}{j-1}.
\end{align}
We do not require explicit expressions for most eigenvalues and valencies of $\Johnson{n}{k}$, the only exceptions being the dual eigenvalues
\begin{equation}
\label{eqn:dualEigenvalues}
q_0(i)=1
\quad\text{and}\quad
q_1(i)
= \frac{n(n-1)}{n-k}\bigg(\frac{k-i}{k}-\frac{k}{n}\bigg).
\end{equation}
See \cite[Eq.~1.24]{RosmanisPhD2014} for the latter.
We are only interested in the following Krein parameters of this association scheme.
When one idempotent is $E_0$, we have
\begin{equation}
\label{eqn:Krein}
q_{i,0}(j)=\Ind_{[i=j]}.
\end{equation}
When one idempotent is $E_1$, we have
\begin{subequations}
\label{eqn:KreinAll}
\begin{align}
\label{eqn:KreinUp}  q_{j-1,1}(j) &=  \frac{j(n-1)n(k-j+1)(m-j+1)}{mk(n-2j+1)(n-2j+2)}, \\
\label{eqn:KreinFlat} q_{j,1}(j) &= \frac{j(n-1)(n-j+1)(m-k)^2}{mk(n-2j)(n-2j+2)},\\
\label{eqn:KreinDown} q_{j+1,1}(j)& =  \frac{n(n-1)(n-j+1)(k-j)(m-j)}{mk(n-2j)(n-2j+1)}, 
\end{align}
\end{subequations}
and $q_{i,1}(j)=0$ whenever $|i-j|>1$ (see \cite[Section 3.2]{BannaiI84}).

\subsection{Adversary lower bound for state discrimination}\label{ssec:advLB}

Consider the following state-discrimination problem.
\begin{itemize}
\item[($\ast$)] Let $f\colon D\to R$ be a function for some finite sets $D$ and $R$.
Let $\{\ket{\psi_x} : x \in D\}$ be a family of quantum states of the same dimension. Given a copy of $\ket{\psi_x}$ for an arbitrary $x \in D$, the goal is to determine $f(x)$ with high success probability.
\end{itemize}
Let $A$ be the Gram matrix of the states, namely
\[
A\elem[x,y] = \braket{\psi_x\;|\;\psi_y},
\]
and let $F$ be the $D\times D$ matrix with
\[
F\elem[x,y] = \Ind_{[f(x)\ne f(y)]}.
\]
Informally, the main result of this section
is that the above state-discrimination problem can be solved with small error probability if and only if 
\[
\gamma_2 (A\circ F)
\]
is small.
We start with the proof of the lower bound.  With constants refined, it reads as follows:

\begin{proposition}
\label{proposition:gamma2Upper}
If the above state-discrimination problem ($\ast$) can be solved with success probability $1-\eps$, then
$
\gamma_2 (A\circ F) \le 4\sqrt{\eps}
$.
\end{proposition}

\begin{proof}
This is essentially the result of~\cite[Claim 3.27]{belovs:phd}, which is also closely related to~\cite{hls:madv}.
For completeness we repeat the proof, with slight modifications.

\newcommand{\Pip}{\Pi^\perp}
Without loss of generality we may assume the measurement is projective (this follows from Neumark's theorem).
Thus, there exist orthogonal projectors $\{\Pi_a\}_{a\in R}$ such that $\norm|\Pi_{f(x)}\ket{\psi_x}|^2\ge 1-\eps$ for all $x\in D$.  
Denote $\Pip_a = \Id - \Pi_a$, so that $\|\Pip_{f(x)}\ket{\psi_x}\|^2\le \eps$ for all $x\in D$. We first write
\begin{align*}
A\elem[x,y]=\braket{\psi_x\;|\;\psi_y} 
&=
 \bra{\psi_x}\Pi_{f(y)}\ket{\psi_y} 
 + \bra{\psi_x} \Pi_{f(y)}^\perp \ket{\psi_y} \\
&=
 \bra{\psi_x} \Pi_{f(x)} \Pi_{f(y)} \ket{\psi_y} 
 + \bra{\psi_x} \Pip_{f(x)} \Pi_{f(y)} \ket{\psi_y} 
 + \bra{\psi_x} \Pip_{f(y)} \ket{\psi_y}
.
\end{align*}
Note that if $f(x)\ne f(y)$, then the first term is 0 because $\Pi_{f(x)}$ and $\Pi_{f(y)}$ project onto orthogonal subspaces.  
This motivates us to define the $D\times D$ matrix
\begin{align*}
B\elem[x,y] 
&= \bra{\psi_x} \Pip_{f(x)} \Pi_{f(y)} \ket{\psi_y} + \bra{\psi_x} \Pip_{f(y)} \ket{\psi_y}.
\end{align*}
We have $A(x,y)=B(x,y)$ whenever $f(x)\ne f(y)$, and hence $A\circ F=B\circ F$. 
Note that $\gamma_2(B) \le 2\sqrt{\eps}$ by taking the vectors $u_x=\big({\eps^{-1/4}} \Pip_{f(x)}\ket{\psi_x}\;,\;\eps^{1/4} \ket{\psi_x}\big)$ and $v_y=\big({\eps^{1/4}}\Pi_{f(y)}\ket{\psi_y} \;,\;\eps^{-1/4}\Pip_{f(y)}\ket{\psi_y}\big)$.
Now we have
\[
\gamma_2(A\circ F) = \gamma_2(B\circ F) \le \gamma_2(B)\gamma_2(F)\le 4\sqrt{\eps},
\]
where we used the composition property of the $\gamma_2$-norm in the first inequality and in the second inequality we used $\gamma_2(F)\le 2$, which follows by considering the vectors $u_x,v_y\in\01^{|R|+1}$ whose last coordinate is always~1, and where $u_x$ has a $1$ at coordinate $f(x)$ and $v_y$ has a $-1$ at coordinate $f(y)$ (identifying $R$ with $\{1,\ldots,|R|\}$ for the purposes of indexing these vectors), and whose remaining entries are all~$0$.
\end{proof}

\begin{proposition}
\label{proposition:gamma2Lower}
The above state-discrimination problem ($\ast$) can be solved with success probability at least $1-\gamma_2(A\circ F)$.
\end{proposition}

\begin{proof}
If $B$ is the Gram matrix of the collection of states $\{\ket{\psi_x}\otimes\ket{f(x)}\}_{x\in D}$, then
\[
A - B = A\circ F.
\]
Using \cite[Claim 3.10]{lee:strongDirect}, there exists a unitary $U$ such that
\[
\bigl(\bra{\psi_x}\otimes\bra{f(x)}\bigr) \, U \, \bigl(\ket{\psi_x}\otimes\ket{0}\bigr) \ge 1-\eps/2
\]
where $\eps := \gamma_2(A\circ F)$.
Thus, if we measure the second register of $U(\ket{\psi_x}\otimes \ket{0})$, we get $f(x)$ with probability at least $(1-\eps/2)^2 \ge 1-\eps$.
\end{proof}



\subsection{Lower bound}
\label{sec:lowersample}

For $x\subseteq [n]$ of size $k$, let
\[
\ket{\psi_x} = \frac1{\sqrt k} \sum_{i\in x}\ket{i}.
\]
This is what we denoted by $\ket{S}$ earlier ($x=S$); we use $\ket{\psi_x}$ here for consistency with the common notation in lower bounds.
The task is to identify the subset $x$ using as few copies of the state $\ket{\psi_x}$ as possible.
We prove the following lower bound.

\begin{theorem}
\label{thm:lowersample}
To find $x$ with success probability $\Omega(1)$, it is necessary to have $\Omega\s[k\log(\min\{k,n-k\})]$ copies of the state $\ket{\psi_x}$.
\end{theorem} 

Let $m = n-k$.
Since we could add more elements to the ambient space artificially, the problem becomes no easier as $n$ grows with $k$ fixed.
Thus, it suffices to prove the lower bound of $\Omega(k\log(m+1))$ under the assumption $m\ll k$.
\medskip

Define the Gram matrix $\Psi$ by 
$
\Psi\elem[x,y] = \braket{\psi_x \; | \; \psi_y}.
$
The Gram matrix corresponding to $\ket{\psi_x}^{\otimes\ell}$ is $\Psi^{\circ\ell}$ (where $\Psi^{\circ \ell}$ is the Hadamard product of $\Psi$ with itself $\ell$ times).
The function we want to compute is $f:x\mapsto x$, so we have $F(x,y)=\Ind_{[f(x)\neq f(y)]}=\Ind_{[x\neq y]}$, i.e., $F=J-\Id$.
By Proposition~\ref{proposition:gamma2Upper}, it thus suffices to prove that for some $\ell = \Omega(k\log(m+1))$ we have
\[
\gamma_2 \sA[ \Psi^{\circ\ell}\circ (J-\Id) ] = \Omega(1).
\]
To that end, we use the dual formulation of the $\gamma_2$-norm (in Eq.~\eqref{eq:dualgamma2}) and construct a matrix $\Gamma$ such that
\[
\|\Gamma\| = 1,\qquad
\Gamma\circ \Id = 0,\qqand
\|\Gamma\circ \Psi^{\circ\ell}\|=\Omega(1).
\]
We now construct a $\Gamma$ that satisfies the constraints above. To do so, we first write~$\Gamma$ in terms of the idempotents $\{E_j\}_{j=0}^m$ of the Johnson association scheme (as defined above Eq.~\eqref{eq:defnoftraceEj}): for $\{\gamma_j\}_j$ which we define shortly, let 
\begin{equation}
\label{eqn:Gamma}
\Gamma = \sum_{j=0}^m \gamma_j E_j.
\end{equation}
To satisfy $\Gamma\circ \Id=0$, we would like $\Gamma$ to have zero diagonal. Note that $\Gamma$ has zero diagonal if and only if $\tr[\Gamma]= \sum_{j=0}^m \gamma_j \tr[E_j]= \sum_{j=0}^m \gamma_j d_j=0$, where $d_j$ was defined in Eq.~\eqref{eq:defnoftraceEj}. We now fix $\{\gamma_j\}_j$ as follows: since  $d_m= \binom{n}{m} - \binom{n}{m-1}$ is larger than the sum of the remaining $d_j$s, we let
\begin{align}
 \label{eq:defnofgamma1tom}   
\gamma_0 = \gamma_1 = \dots = \gamma_{m-1}=1,
\qquad
\gamma_{m}\in[-1,0]
\end{align}
so that $\tr[\Gamma]=0$ and $\|\Gamma\|=1$.
Thus, it remains to show that 
\begin{equation}
\label{eqn:whatWeNeed}
\|\Gamma\circ \Psi^{\circ\ell}\|=\Omega(1).
\end{equation}
For that, we use the following technical result.

\begin{lemma}
\label{lem:circPsi}
For each $j=0,1,\dots,m$, we have
\[
E_j\circ\Psi = p_{j+1,-1} E_{j+1} + p_{j,0}E_j + p_{j-1,+1}E_{j-1},
\]
where
\begin{align*}
p_{j,-1} &= \frac{j(k-j+1)(m-j+1)}{(n-2j+1)(n-2j+2)k}, \\
p_{j,0}  &= \frac{k}{n} + \frac{j(n-j+1)(m-k)^2}{nk(n-2j)(n-2j+2)},\\
p_{j,+1} &= \frac{(n-j+1)(k-j)(m-j)}{(n-2j)(n-2j+1)k}.
\end{align*}
\end{lemma}

Before we proceed with the proof of this lemma, let us state a simple consequence.

\begin{corollary}
For each $j\in\{0,\dots,m\}$, the numbers $p_{j,-1}$, $p_{j,0}$, and $p_{j,+1}$ are non-negative, and satisfy $p_{j,-1} + p_{j,0} + p_{j,+1} = 1$.
\end{corollary}

\begin{proof}
The non-negativity is obvious.
For the last property note that
\begin{equation*}
\sum_{j=0}^m E_j = \Id = \Psi\circ \Id = \Psi\circ\sC[\sum_{j=0}^m E_j] = \sum_{j=0}^m \sA[p_{j,-1} + p_{j,0} + p_{j,1}] E_j, 
\end{equation*}
where the first equality uses the definition  of an association scheme, the second equality follows because $\Psi\elem[x,x]=1$ by definition, and the last equality is by the assumption of Lemma~\ref{lem:circPsi}.
\end{proof}

\begin{proof}[\bfseries \upshape Proof of Lemma~\ref{lem:circPsi}]
It suffices to write out $\Psi$ in the basis $\{E_j\}_{j=0}^m$ and use the Krein parameters.
By definition of $\ket{\psi_x}=\frac{1}{\sqrt{k}}\sum_{i\in x}\ket{i}$, we have that $\Psi\elem[x,y]=\braket{\psi_x \; | \; \psi_y}$ equals $\frac{1}{k}$ times the intersection of~$x$ and $y$, and
\[
\Psi = \sum_{i=0}^m \sB[1-\frac ik] A_i,
\]
where $A_i$ was defined at the beginning of Section~\ref{ssec:johnson} as $A_i\elem[x,y]:=\Ind_{[|x\cap y|=k-i]}$.
We now rewrite $\Psi$ as follows: using Eq.~\eqref{eqn:dualEigenvalues}, 
we have
\[
\frac kn E_0 + \frac{n-k}{n(n-1)} E_1 = \frac1N \sum_{i=0}^m \sB[ \frac kn q_0(i) + \frac{n-k}{n(n-1)} q_1(i)] A_i = \frac1N\sum_{i=0}^m \frac {k-i}k A_i = \frac1N\Psi,
\]
where the first equality used Eq.~\eqref{eq:eigenvaluesanddual}. 
Additionally observe that
\[
N E_j\circ E_0 = q_{j,0}(j) E_j
\qqand
N E_j\circ E_1 = q_{j,1}(j-1) E_{j-1} + q_{j,1}(j) E_j + q_{j,1}(j+1) E_{j+1}.
\]
Plugging in the values of $q_{j,\cdot}$ from Eq.~\eqref{eqn:KreinAll}, we get the required equality.
\end{proof}

We are now ready to prove our main lower bound in Theorem~\ref{thm:lowersample}. 

\begin{proof}[\bfseries \upshape Proof of Theorem~\ref{thm:lowersample}]
We prove this by induction on the number of copies of the state $\ket{\psi_x}$, which we denote by $s$.
Let us define $\gamma^{(s)}_j$ via
\[
\Gamma\circ\Psi^{\circ s} = \sum_{j=0}^m \gamma^{(s)}_j E_j.
\]
Since the $E_j$ are pairwise-orthogonal projections, the norm of $\Gamma\circ\Psi^{\circ s}$ equals $\max_j |\gamma^{(s)}_j|$.
Hence to lower bound $\|\Gamma\circ\Psi^{\circ s} \|$, it suffices to lower bound $\gamma^{(s)}_0$.

We have
\[
\Gamma \circ \Psi^{\circ (s+1)} = \sum_{j=0}^m \gamma_j^{(s)} E_j \circ \Psi
\]
and using Lemma~\ref{lem:circPsi} we get
\begin{equation}
\label{eqn:gamma_j^{(s+1)}}
\gamma_j^{(s+1)} = p_{j,-1} \gamma_{j-1}^{(s)} + 
p_{j,0} \gamma_{j}^{(s)} +
p_{j,+1} \gamma_{j+1}^{(s)}.
\end{equation}
For every $j\in \{0,\ldots,m\}$, we now consider the following probabilistic sequence  $\{B_j^{(s)}\}$.
For $s=0$, we let $B_j^{(0)} = \gamma_j$ and
\[
B_j^{(s+1)} =
\begin{cases}
B_{j-1}^{(s)}&\text{with probability $p_{j,-1}$},\\
B_{j}^{(s)}&\text{with probability $p_{j,0}$},\\
B_{j+1}^{(s)}&\text{with probability $p_{j,+1}$},
\end{cases}
\]
using the fact that $p_{j,-1} + p_{j,0} + p_{j,+1} = 1$. 
Note that $B_j^{(s)}$ only takes values from $\{\gamma_0,\ldots,\gamma_m\}$ and there are only two distinct such values, namely 1 and $\gamma_m$ (since $\gamma_0 = \gamma_1 = \dots = \gamma_{m-1}=1$ as defined in Eq.~\eqref{eq:defnofgamma1tom}).
Also note that $p_{0,-1}=p_{m,+1}=0$, so we do not have to explicitly handle the boundaries.
Induction on~$s$ using Eq.~\eqref{eqn:gamma_j^{(s+1)}} shows that $\Exp[B_j^{(s)}] = \gamma_j^{(s)}$, which is the motivation behind defining these variables.

Define similarly $C_j^{(s)}$ as $C_j^{(0)} = \gamma_j$ and
\[
C_j^{(s+1)} =
\begin{cases}
C_{j}^{(s)}&\text{with probability $p_{j,-1}+p_{j,0}$},\\
C_{j+1}^{(s)}&\text{with probability $p_{j,+1}$}.
\end{cases}
\]
Let us give an intuitive description of how the random variables $C_j^{(s)}$ behave.
For each $s$, the head of the sequence $C^{(s)}_0, C^{(s)}_1,\dots,$ up to some $ C^{(s)}_\ell$ consists purely of 1s, and the tail $C^{(s)}_{\ell+1},\dots,C^{(s)}_m$ consists purely of $\gamma_m$.
Initially, for $s=0$, the tail consists of one element $C^{(s)}_m$ only, but the tail gradually extends as $s$ grows (and the head, respectively, shrinks).
The probability of growing the length of the tail from $m-j$ to $m-j+1$ in one step is $p_{j,+1}$.

The random variables $B_j^{(s)}$ behave similarly, but are slightly more complicated, since the tail can also shrink and 1s can get into the tail.
This is the reason why we replace $B_j^{(s)}$ with $C_j^{(s)}$ in our analysis: $C_j^{(s)}$ is easier to analyze, and it suffices to lower bound its expectation because $B^{(s)}_j$ dominates $C^{(s)}_j$, i.e., for each $s$ and $j$ and real $t$ we have $\Pr[B^{(s)}_j\ge t]\ge \Pr[C^{(s)}_j\ge t]$.
The latter is proven by induction, as follows.  The base case $s=0$ is trivial, and the inductive step is
\begin{align*}
\Pr[B^{(s+1)}_j\ge t] &=
p_{j,-1} \Pr[B^{(s)}_{j-1}\ge t] +
p_{j,0} \Pr[B^{(s)}_{j}\ge t] +
p_{j,+1} \Pr[B^{(s)}_{j+1}\ge t]\\
&\ge
p_{j,-1} \Pr[C^{(s)}_{j-1}\ge t] +
p_{j,0} \Pr[C^{(s)}_{j}\ge t] +
p_{j,+1} \Pr[C^{(s)}_{j+1}\ge t]\\
&\ge
(p_{j,-1}+p_{j,0}) \Pr[C^{(s)}_{j}\ge t] +
p_{j,+1} \Pr[C^{(s)}_{j+1}\ge t]
= \Pr[C^{(s+1)}_j\ge t],
\end{align*}
since $C^{(s)}_{j-1}\ge C^{(s)}_j$ by our above analysis.

The analysis of $C^{(s)}_j$ is very similar to the classical coupon collector problem if we interpret
the length of the tail as the number of acquired coupons.
We briefly repeat the argument.
For each $j$, define random variable $T_j$ as the first value of $s$ such that $C_j^{(s)}=\gamma_m$.
Obviously, $T_m = 0$.
We can interpret~$T_j$ as the first value of $s$ such that the length of the tail becomes $m-j+1$.
The random variable $T_j - T_{j+1}$ is the number of steps required to grow the length of the tail from $m-j$ to $m-j+1$.
Clearly, these variables are independent for different $j$.
Also, each of them is distributed according to a geometric distribution and standard probability theory gives us that 
$\Exp[T_j - T_{j+1}] = 1/p_{j,+1}$
and
$\Var[T_j - T_{j+1}] = (1-p_{j,+1})/p^2_{j,+1}$.
%
%
We have $p_{j,+1} = \Theta((m-j)/k)$ from Lemma~\ref{lem:circPsi}, so
\[
\Exp[T_0] = \sum_{j=0}^{m-1} \frac 1{p_{j,+1}} = \Theta(k)\sC[\sum_{j=0}^{m-1} \frac 1{m-j}] = \Theta(k\log(m+1)).
\]
Similarly,
\[
\Var[T_0] = \sum_{j=0}^{m-1} \frac{1-p_{j+1}}{p_{j,+1}^2} = \Theta(k^2) \sC[\sum_{j=0}^{m-1} \frac 1{(m-j)^2}] = \Theta(k^2).
\]
Hence, using Chebyshev's inequality, there exists $\ell = \Theta(k\log(m+1))$ such that
\[
\Pr[T_0 > \ell] \ge 3/4.
\]
Since $C_0^{(\ell)}$ can take only two values (1 and $\gamma_m\in[-1,0]$), we have that
\[
\gamma_{0}^{(\ell)} = \Exp[C_0^{(\ell)}] \ge 3/4\cdot 1 + 1/4\cdot \gamma_m \ge 1/2.
\]
Finally, since $B^{(\ell)}_0$ dominates $C^{(\ell)}_0$, we get
\[
\gamma_{0}^{(\ell)} = \Exp[B_0^{(\ell)}] \ge \Exp[C_0^{(\ell)}] \ge 1/2,
\]
implying Eq.~\eqref{eqn:whatWeNeed}.
This shows that there exists $\ell = \Theta(k\log(m+1))$ such that the error probability of any measurement on $\ell$ copies of $\ket{\psi_x}$ has error probability $\Omega(1)$ in identifying~$x$.
\end{proof}

\section{Learning from quantum samples and reflections}\label{sec:numberofreflections}

In the previous sections we assumed we were given a number of copies of the unknown state~$\ket{S}$.
In this section we assume a stronger model: in addition to a number of copies of the state~$\ket{S}$, we are also given the ability to apply the reflection $R_S=2\ketbra{S}{S}-\Id$ through $\ket{S}$.  The key additional tool we will use is (exact) amplitude amplification, encapsulated by the next theorem, which follows from~\cite{bhmt:countingj}:

\begin{theorem}[Exact amplitude amplification]
Let $\ket{\phi}$ and $\ket{\psi}$ be states such that $\inpc{\phi}{\psi}=\alpha>0$. Suppose we know $\alpha$ exactly, and we can implement reflections through $\ket{\phi}$ and $\ket{\psi}$.
Then we can convert $\ket{\phi}$ into $\ket{\psi}$ (exactly) using $O(1/\alpha)$ reflections and $\widetilde{O}(1/\alpha)$ other gates.
\end{theorem}

\noindent
We distinguish the two regimes of $k\geq n/2$ and $k<n/2$.

\subsection{Tight bound if \texorpdfstring{$k\geq n/2$}{k>=n/2}}


\begin{theorem}[Upper bound for small $m$]\label{thm:uppersamplerefsmallm}
Let $S\subseteq [n]$ be a set of size $k\geq n/2$ and let $m=n-k$. We can identify $S$ with probability $1$ using $O\left(\sqrt{km}\right)$ uses of $R_S=2\ketbra{S}{S}-\Id$.
\end{theorem}

\begin{proof}
Our algorithm sequentially finds all $m$ missing elements.
We would like to use amplitude amplification to prepare a copy of $\ket{\overline{S}}$, which is the uniform superposition over the $m$ missing elements.
Consider the uniform state over the $n$-element universe:
$$
\ket{[n]}=\sqrt{\frac{k}{n}}\ket{S} + \sqrt{\frac{m}{n}}\ket{\overline{S}}.
$$
This state is easy to prepare, and hence also easy to reflect through.
Note that in the 2-dimensional plane spanned by $\ket{S}$ and $\ket{\overline{S}}$, reflection through $\ket{\overline{S}}$ is the same as a reflection through $\ket{S}$ up to an irrelevant global phase.
The inner product between $\ket{[n]}$ and $\ket{\overline{S}}$ equals $\sqrt{m/n}$. Accordingly, using $O(\sqrt{n/m})$ rounds of exact amplitude amplification (which only rotates in the 2-dimensional space spanned by $\ket{S}$ and $\ket{\overline{S}}$; each round ``costs'' one application of $R_S$) we can turn $\ket{[n]}$ into $\ket{\overline{S}}$, up to a global phase.

Measuring $\ket{\overline{S}}$ gives us one of the missing elements, uniformly at random. Now we remove this element from the universe. Note that $\ket{S}$ does not change since we removed an element of the universe that was missing from $S$. We then repeat the above algorithm on a universe of size $n-1$ with $m-1$ missing elements in order to find another missing element at the cost of $O(\sqrt{(n-1)/(m-1)})$ rounds of amplitude amplification, and so on.
This finds all missing elements (and hence $S$) with probability~1, using
$$
\sum_{i=0}^{m-1} O\left(\sqrt{\frac{n-i}{m-i}}\right)=O(\sqrt{n})\sum_{j=1}^{m} \frac{1}{\sqrt{j}}=O(\sqrt{nm})=O(\sqrt{km})
$$
applications of $R_S$, where we used $k\geq n/2$. Note that in this regime we do not need any copies of~$\ket{S}$, just reflections $R_S$.
\end{proof}

\begin{theorem}[Lower bound for small $m$]\label{thm:lowersamplerefsmallm}
Let $S\subseteq [n]$ be a set of size $k<n$ and let $m=n-k$. Any quantum algorithm that identifies $S$ with high probability using a total of $T$ copies of $\ket{S}$ and uses of~$R_S$, must satisfy $T=\Omega\left(\sqrt{km}\right)$. The lower bound holds even if we allow $T$ copies of $\ket{S}$, uses of $R_S$, and membership queries to $S$.
\end{theorem}
\begin{proof}
We prove a matching lower bound in a stronger model, namely in a model where we can make queries to the $n$-bit characteristic vector~$x$ for~$S$. That is, we now assume we have a unitary~$U_S$ that maps
$$
U_S:\ket{i,b}\mapsto\ket{i,b\oplus x_i}\mbox{~~~for all }i\in[n],b\in\01, 
$$
where $x_i=1$ iff $i\in S$.

We first argue that this is indeed a stronger model, by showing how we can unitarily prepare a copy of $\ket{S}$ using $O(1)$ applications of $U_S$.
Note that $\inpc{[n]}{S}=\sqrt{k/n}\geq 1/\sqrt{2}$ under the current assumption that $k\geq n/2$. Also note that, in the 2-dimensional space spanned by $\ket{S}$ and $\ket{\overline{S}}$, a reflection through $\ket{S}$ corresponds to a ``phase query'' to~$x$, which can be implemented by one query to $U_S$ (setting the target qubit to $(\ket{0}-\ket{1})/\sqrt{2}$). Hence using $O(1)$ rounds of exact amplitude amplification suffices to prepare a copy of $\ket{S}$ starting from the state $\ket{[n]}$, which is easy to prepare and reflect through. Thus we can implement the state-preparation map $G_S\colon\ket{0}\mapsto \ket{S}$ using $O(1)$ applications of $U_S$. Note that one application of $G_S^{-1}$, followed by a reflection through $\ket{0}$ and an application of $G_S$, implements a reflection through $\ket{S}$. Thus preparing a copy of $\ket{S}$ and reflecting through $\ket{S}$ each ``cost'' only $O(1)$ queries to~$x$ (i.e., applications of $U_S$).

Accordingly, an algorithm that learns $S$ using at most $T$ copies of $\ket{S}$ and at most $T$ applications of $R_S$ implies a quantum algorithm that can learn an $n$-bit string $x$ of weight $k\geq n/2$ using $O(T)$ queries to~$x$.
But it is known that this requires $\Omega(\sqrt{nm})=\Omega(\sqrt{km})$ queries to~$x$, even when allowing bounded error probability. This follows, for instance, from~\cite[Theorem~4.10]{bbcmw:polynomialsj}. Hence we obtain the same lower bound on the number of copies of $\ket{S}$ plus the number of reflections through~$\ket{S}$.
\end{proof}

\subsection{Tight bound if \texorpdfstring{$k<n/2$}{k<n/2}}


\begin{theorem}[Upper bound for small $k$]\label{thm:uppersamplerefsmallk}
Let $S\subseteq [n]$ be a set of size $k<n$. We can identify $S$ with probability $1$ using $O(k)$ copies of $\ket{S}$ and uses of $R_S=2\ketbra{S}{S}-\Id$.
\end{theorem}

\begin{proof}
Our algorithm sequentially finds all elements of $S$.
We start with a copy of $\ket{S}$ and measure to find one $i_1\in S$. Then we use exact amplification to convert a fresh copy of $\ket{S}$ into $\ket{S\setminus\{i_1\}}$. This requires being able to reflect through $\ket{S}$ (i.e., apply $R_S$), and reflect through $\ket{S\setminus\{i_1\}}$. In the 2-dimensional plane spanned by $\ket{S}$ and $\ket{S\setminus\{i_1\}}$, the latter reflection is equivalent to putting a minus in front of $\ket{i_1}$, which is easy to do. 
We measure $\ket{S\setminus\{i_1\}}$ and learn (with probability~1) another element $i_2\in S\setminus\{i_1\}$. Then we change a fresh copy of $\ket{S}$ into $\ket{S\setminus\{i_1,i_2\}}$, measure, and learn some $i_3\in S\setminus\{i_1,i_2\}$.
We repeat this until we have seen all $k$ elements.

The amplitude amplifications get more costly as we find more elements of $S$: If we have already found a set $I\subseteq S$, then changing a fresh copy of $\ket{S}$ to $\ket{S\setminus I}$ uses $O(\frac{1}{\inpc{S}{S\setminus I}})=O(\sqrt{k/(k-|I|)})$ reflections, and hence $O(\sqrt{k/(k-|I|)})$ applications of~$R_S$. Overall, this procedure finds $S$ using $k=|S|$ copies of $\ket{S}$, and
$$
\sum_{i=0}^{k-1} O\left(\sqrt{\frac{k}{k-i}}\right)=O(\sqrt{k})\sum_{j=1}^k \frac{1}{\sqrt{j}}=O(k)
$$
applications of $R_S$.
\end{proof}

\begin{theorem}[Lower bound for small $k$]\label{thm:lowersamplerefsmallk}
Let $S\subseteq [n]$ be a set of size $k<n$. Any quantum algorithm that identifies $S$ with high probability using a total of $T$ copies of $\ket{S}$, and uses of $R_S$ must satisfy $T=\Omega(k)$. The lower bound holds even if we allow $T$ copies of $\ket{S}$, uses of $R_S$, and membership queries to $S$.
\end{theorem}

\begin{proof}
To prove a matching lower bound, suppose our algorithm receives advice in the form of $n-2k$ of the missing elements. This advice reduces the problem to one with universe size $n'=n-(n-2k)=2k$ and $m'=m-(n-2k)=k$ missing elements. Importantly, note that $\ket{S}$, and hence $R_S$, do not change after learning these missing elements. But in Theorem \ref{thm:lowersamplerefsmallm} we already proved an $\Omega(\sqrt{n'm'})=\Omega(k)$ lower bound on the number of copies of $\ket{S}$, reflections, and queries to $S$ needed to solve this special case.  Since the extra advice cannot have made the original problem harder, the same lower bound applies to our original problem.
\end{proof}

\section{Proper PAC learning}\label{sec:properPAC}

As mentioned briefly in the introduction, one of the motivations for this research is the question whether the sample complexity of \emph{proper} quantum PAC learning is higher than that of improper PAC learning.
Let us precisely define Valiant's PAC model~\cite{valiant:paclearning}. We are trying to learn an unknown element $f$ from a \emph{concept class} $\Cc$. For simplicity we only consider $f$s that are Boolean-valued functions on~$[n]$. Our access to~$f$ is through random examples, which are pairs of the form $(x,f(x))$, where $x$ is distributed according to a distribution $D\colon[n]\rightarrow [0,1]$ that is unknown to the learner. A learning algorithm takes a number $T$ of such i.i.d.\ examples as input, and produces a hypothesis~$h\colon[n]\rightarrow \01$ that is supposed to be close to the target function~$f$. The error of the hypothesis $h$ (with respect to the target~$f$, under distribution~$D$) is defined as
$$
\err_D(f,h):=\Pr_{x\sim D}[f(x)\neq h(x)].
$$
We say that a learning algorithm is an \emph{$(\eps,\delta)$-PAC learner for $\Cc$}, if it probably (i.e., with probability at least $1-\delta$) outputs an approximately correct (i.e., with error at most $\eps$) hypothesis $h$:
$$
\forall f\in \Cc,\forall D: \Pr[\err_D(f,h)>\eps]\leq\delta,
$$
where the probability is taken over the sequence of $T$ $D$-distributed examples that the learner receives, as well as over its internal randomness.
The \emph{$(\eps,\delta)$-PAC sample complexity} of~$\Cc$ is the minimal $T$ for which such a learning algorithm exists.\footnote{This definition uses the information-theoretic notion of \emph{sample} complexity. We do not consider the \emph{time} complexity of learning here. For more on sample and time complexity of quantum learning, we refer the reader to~\cite{arunachalam:quantumlearningsurvey}.}

The PAC sample complexity of $\Cc$ is essentially determined by its VC-dimension~$d$ as\footnote{The VC-dimension of $\Cc$ is the maximum size among all sets~$T\subseteq[n]$ that are ``shattered'' by $\Cc$. A set $T$ is shattered by $\Cc$ if for all $2^{|T|}$ labelings $\ell\colon T\to\01$ of the elements of~$T$, there is an $f\in\Cc$ that has that labeling (i.e., where $f_{|T}=\ell$).} 
\begin{equation}\label{eq:VCsample}
\Theta\left(\frac{d}{\eps} + \frac{\log(1/\delta)}{\eps}\right).
\end{equation}
See Blumer et al.~\cite{blumer:optimalpacupper} for the lower bound and Hanneke~\cite{hanneke:optimalpac} for the upper bound. 

The above upper bound on sample complexity allows the learner to be improper, i.e., to sometimes output hypotheses $h\not\in\Cc$.
The following folklore example, which we learned from Steve Hanneke \cite{hanneke:comm}, shows that the sample complexity of \emph{proper} learning can be asymptotically larger.\footnote{In a recent result, Montasser et al.~\cite{mhs:robust} proved another separation between proper and improper learning.} 
Consider the concept class $\Cc=\{f\colon[n]\to\01 \mid \exists!~i\text{ s.t.\ }f(i)=0\}$ of functions that are all-1 except on one ``missing element''~$i$.
The VC-dimension of this class is~1, hence $\Theta\bigl(\frac{\log(1/\delta)}{\eps}\bigr)$ classical examples are necessary and sufficient for PAC learning~$\Cc$ by \eqref{eq:VCsample}.
With $\eps=1/n$ and $\delta=1/3$, this bound becomes $\Theta(n)$.
Now fix an $(\eps,\delta)$-PAC \emph{proper} learner for this class that uses some $T$ examples; 
we will show that $T=\Omega(n\log n)$, exhibiting an asymptotic separation between the sample complexities of proper and improper PAC learning.

For every $i\in [n]$, consider a distribution $D_i$ that is uniform over $[n]\setminus\{i\}$. If the target concept~$f$ has $i$ as its missing element then the learner has to output that $f$, since any other $g\in\Cc$ will make an error on its own missing element and hence would have error at least $1/(n-1)>\eps$ under~$D_i$. 
In other words, when sampling from $D_i$ the learner has to identify the one missing element~$i$ with success probability $\geq 2/3$. 
But we know from the coupon collector argument that this requires $\Omega(n\log n)$ samples. Note that a $D_i$-distributed $(x,f(x))$ is equivalent to sampling uniformly from $[n]\setminus\{i\}$, since the label~$f(x)$ is always~1 under~$D_i$.

What about \emph{quantum} PAC learning? Bshouty and Jackson~\cite{bshouty:quantumpac} generalized the PAC model by considering superposition states
$$
\ket{\psi_{D,f}}=\sum_x\sqrt{D(x)}\ket{x,f(x)}
$$
instead of random samples. The learner now receives $T$ copies of this ``quantum example'' state, and has to output a probably approximately correct hypothesis.
Measuring a quantum example gives a classical example, so quantum examples are at least as useful as classical examples, but one of the questions in quantum learning theory is in what situations they are significantly more useful.
Two of us~\cite{arunachalam:optimalpaclearning} have shown that the bound of \eqref{eq:VCsample} also applies to learning from quantum examples, so for improper learning the quantum and classical sample complexities are equal up to constant factors. However, quantum examples \emph{are} beneficial for learning $\Cc$ under the $D_i$ distributions. Note that $\ket{\psi_{D_i,f_i}}$ is just the uniform superposition over the set $S=[n]\setminus\{i\}$, tensored with an irrelevant extra $\ket{1}$. As we showed in Section~\ref{sec:upperbound}, given $O(n)$ copies of $\ket{\psi_{D_i,f_i}}$ we can identify the one missing element~$i$ with probability $\geq 2/3$. So the example that separates the sample complexities of \emph{classical} proper and improper learning, does not separate \emph{quantum} proper and improper learning. This naturally raises the question of whether the quantum sample complexities of proper and improper PAC learning are asymptotically equal (which, as mentioned, they provably are not in the classical case).


\paragraph{Acknowledgments.}
We thank Steve Hanneke for sharing with us the folklore example separating classical proper and improper PAC learning~\cite{hanneke:comm}.
AMC, RK, and RdW did part of this work during the ``Challenges in Quantum Computation'' program at the Simons Institute for the Theory of Computing in Berkeley in Summer 2018 and gratefully acknowledge its hospitality.

\bibliographystyle{alpha}
\bibliography{qcs}

\end{document}